\documentclass[conference]{IEEEtran}
\makeatletter
\def\ps@headings{%
\def\@oddhead{\mbox{}\scriptsize\rightmark \hfil \thepage}%
\def\@evenhead{\scriptsize\thepage \hfil \leftmark\mbox{}}%
\def\@oddfoot{}%
\def\@evenfoot{}}
\makeatother
\ifCLASSINFOpdf
  \usepackage[pdftex]{graphicx}
\else
  \usepackage[dvips]{graphicx}
\fi
\usepackage{color}

\usepackage{amsmath}

\usepackage{algpseudocode}
\usepackage{algorithm}
\usepackage{caption}
\usepackage{subcaption}

\usepackage{amsthm}
\newtheorem{theorem}{Theorem}
\newtheorem{theo}{Theorem}
\newtheorem{definition}[theo]{Definition}

\begin{document}
%
% paper title
% can use linebreaks \\ within to get better formatting as desired
\title{Technical Report: Sleep-Route -- Routing through Sleeping Sensors}

% author names and affiliations
% use a multiple column layout for up to three different
% affiliations
\author{\IEEEauthorblockN{Chayan Sarkar, Vijay S. Rao, R. Venkatesha Prasad}
\IEEEauthorblockA{Embedded Software Group\\
Delft University of Technology, The Netherlands\\
\{C.Sarkar, V.Rao, R.R.VenkateshaPrasad\}@tudelft.nl
}}

% conference papers do not typically use \thanks and this command
% is locked out in conference mode. If really needed, such as for
% the acknowledgment of grants, issue a \IEEEoverridecommandlockouts
% after \documentclass

% for over three affiliations, or if they all won't fit within the width
% of the page, use this alternative format:
% 
%\author{\IEEEauthorblockN{Michael Shell\IEEEauthorrefmark{1},
%Homer Simpson\IEEEauthorrefmark{2},
%James Kirk\IEEEauthorrefmark{3}, 
%Montgomery Scott\IEEEauthorrefmark{3} and
%Eldon Tyrell\IEEEauthorrefmark{4}}
%\IEEEauthorblockA{\IEEEauthorrefmark{1}School of Electrical and Computer Engineering\\
%Georgia Institute of Technology,
%Atlanta, Georgia 30332--0250\\ Email: see http://www.michaelshell.org/contact.html}
%\IEEEauthorblockA{\IEEEauthorrefmark{2}Twentieth Century Fox, Springfield, USA\\
%Email: homer@thesimpsons.com}
%\IEEEauthorblockA{\IEEEauthorrefmark{3}Starfleet Academy, San Francisco, California 96678-2391\\
%Telephone: (800) 555--1212, Fax: (888) 555--1212}
%\IEEEauthorblockA{\IEEEauthorrefmark{4}Tyrell Inc., 123 Replicant Street, Los Angeles, California 90210--4321}}

% use for special paper notices
%\IEEEspecialpapernotice{(Invited Paper)}

% make the title area
\maketitle

\begin{abstract}
In this article, we propose an energy-efficient data gathering scheme for wireless sensor network called \emph{Sleep-Route}, which splits the sensor nodes into two sets -- active and dormant (low-power sleep). Only the active set of sensor nodes participate in data collection. The sensing values of the dormant sensor nodes are predicted with the help of an active sensor node. Virtual Sensing Framework (VSF) provides the mechanism to predict the sensing values by exploiting the data correlation among the sensor nodes. If the number of active sensor nodes can be minimized, a lot of energy can be saved. The active nodes' selection must fulfill the following constraints - (i)~the set of active nodes are sufficient to predict the sensing values of the dormant nodes, (ii)~each active sensor nodes can report their data to the sink node (directly or through some other active node(s)). The goal is to select a minimal number of active sensor nodes so that energy savings can be maximized.

The optimal set of active node selection raise a combinatorial optimization problem, which we refer as \emph{Sleep-Route} problem. We show that Sleep-Route problem is NP-hard. Then, we formulate an integer linear program (ILP) to solve the problem optimally. To solve the problem in polynomial time, we also propose a heuristic algorithm that performs near optimally.
\end{abstract}

% IEEEtran.cls defaults to using nonbold math in the Abstract.
% This preserves the distinction between vectors and scalars. However,
% if the conference you are submitting to favors bold math in the abstract,
% then you can use LaTeX's standard command \boldmath at the very start
% of the abstract to achieve this. Many IEEE journals/conferences frown on
% math in the abstract anyway.

% no keywords

% For peer review papers, you can put extra information on the cover
% page as needed:
% \ifCLASSOPTIONpeerreview
% \begin{center} \bfseries EDICS Category: 3-BBND \end{center}
% \fi
%
% For peerreview papers, this IEEEtran command inserts a page break and
% creates the second title. It will be ignored for other modes.
\IEEEpeerreviewmaketitle
%
%vazifehdan2012analytical, prasad2014reincarnation,
%
\section{Background}
\label{background}
Since Sleep-Route works in conjunction with VSF, we provide an overview of VSF in this section. Every sensor node in a WSN senses a physical parameter at a predefined interval and transmits this data to the sink node. Usually, the data collected from various sensors show correlation among themselves -- some sensors are highly correlated while some are less correlated. Let us take the example in Fig.~\ref{fig:wsn}, if two sensors a and d are reporting highly correlated data, the data of d can be predicted (with high accuracy) from the data of a, without actually sensing the physical parameter by d. d can be kept dormant and can save energy. VSF proposes to save energy by exploiting this correlation for wireless sensor nodes \cite{sarkar2013nosense}. This situation can occur when sensors use energy harvesting wherein some nodes harvest energy for a short duration and die after using that energy\cite{prasad2014reincarnation}. To predict the sensor data, a virtual sensor (VS) is created for each physical sensor (PS) at the sink as shown in Fig.~\ref{fig:wsn}. A VS instructs its associated PS on its state and activity for the forthcoming sensing intervals. 

From the datasets collected from real-deployment~\cite{nrel,sensorscope,madden2004}, it is found that physical parameter values (e.g., ambient temperature, humidity, etc.) do not change abruptly over a short time span. Therefore, these values have correlation with their immediate past values (temporal correlation). Hence, VSF exploits temporal as well as spatial correlations in the data collected from the sensors to fine-tune the prediction. In active mode, all the circuits of a sensor node such as microprocessor, radio, sensor, clock, etc., remains switched on. As a result, the sensor node spends higher energy in active mode. On the other hand, most of the circuits of a sensor node remain switched off in dormant state, and the energy consumption is minimal. To perform any operation, a sensor node has to be in active mode and it is imperative that number of active nodes should be minimized. In order to keep the sensor network alive for a longer period, energy consumption of the sensor nodes in the network needs to be balanced over time. Thus, state of the nodes switch between dormant and active after certain number of sensing intervals. 
\begin{figure}
\centering
	\begin{subfigure}[b]{0.42\linewidth}
		\includegraphics[width=\textwidth]{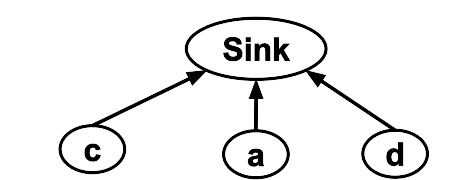}
		\caption{}
		\label{fig:scenario1}
	\end{subfigure}
\hspace{2em}
	\begin{subfigure}[b]{0.42\linewidth}
		\includegraphics[width=\textwidth]{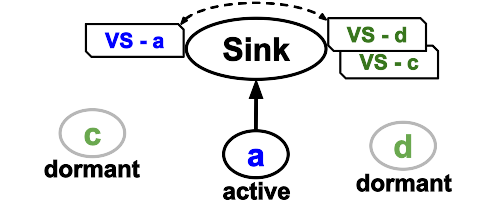}
		\caption{}
		\label{fig:scenario2}
	\end{subfigure}
\caption{(a) Data collection scenario in a WSN (star topology); (b) Data collection with virtual sensing framework.} 
\label{fig:wsn}
\end{figure}

The VSF does not assume any \textit{a priori} knowledge about the sensor data statistics. It captures the correlation amongst sensors on the fly, monitors the change in the correlation, and adapts dynamically. To accomplish this, the data collection period is divided into three phases -- training period, operational period and revalidation period as shown in Fig.~\ref{fig:vs_train}. 
\begin{figure}
\centering
\includegraphics[width=\linewidth]{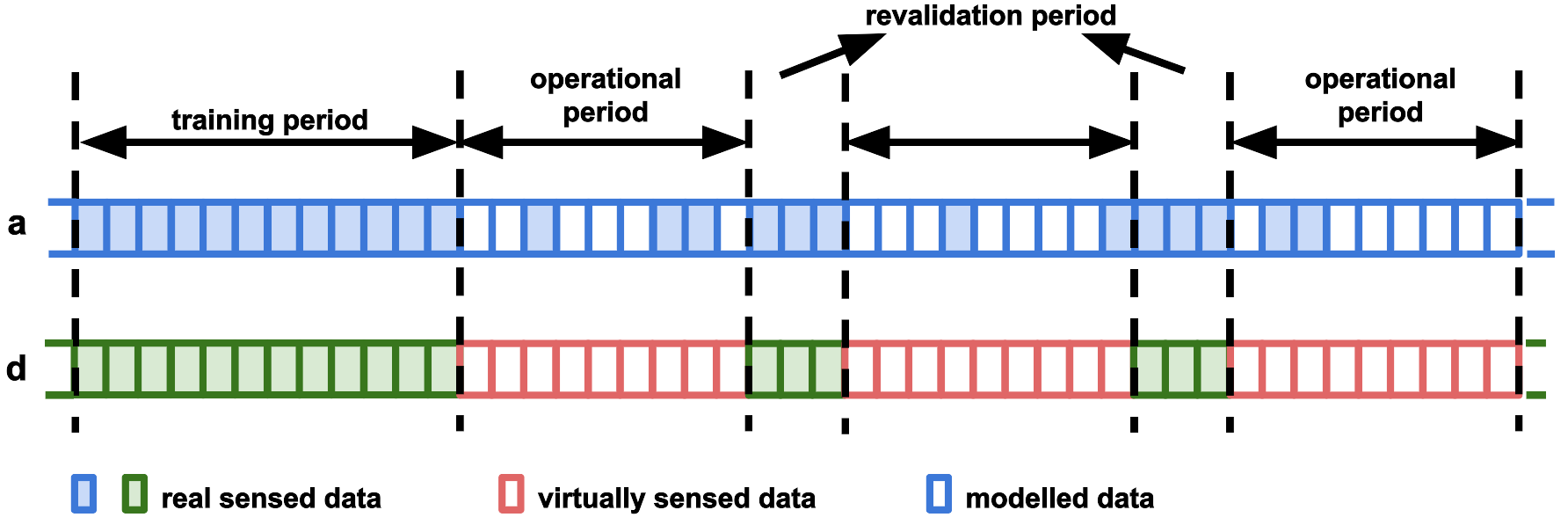}
\caption{Data collection phases in virtual sensing technique.}
\label{fig:vs_train}
\end{figure}

During the training period, all nodes remain active and report their sensed data for certain sensing intervals (data samples) to the sink node. Using these samples, the sink node creates a temporal correlation based predictor for each node. The temporal predictor used by the VSF is a transversal filter, which is created by mean-square error method. Further, these training samples are exploited to find correlated node pairs. For each node in a pair, a separate spatial correlation based predictor is created. VSF uses linear regression to find a spatial predictor. 

At the beginning of the operational period, each sensor is informed about its operating state (i.e. active/dormant) by the sink node. If a node becomes dormant during the operational period, the virtual sensor uses a hybrid filter (combination of the temporal and spatial predictor) to predict the sensed value.
It is hard, if not impossible, to find the accuracy of the predicted values at VS of a dormant sensor vis-\'a-vis the values if the sensor would have been active. To tackle this, the dormant nodes become active again for a small duration after each operational period. During this period, called revalidation period, all the sensor nodes become active and report their sensed data to the sink node. Then, the sink node checks the correlations amongst the sensors and readjust the parameters of the prediction model for the virtual sensors if required. This ensures higher prediction accuracy in the next operational period. If the sink node identifies a large drift in the correlation pattern -- in turn predicted values -- during revalidation, another training period is initiated. Otherwise, before resuming another operational period after the revalidation period, the sink node picks a new set of active nodes in order to balance the energy expenditure by each sensor node.

\section{Problem Description}
\label{sec:problem}
In a large WSN, the VSF provides a mechanism to find correlations amongst nodes, and classify highly correlated nodes into ``groups''. We refer to such a group as \emph{constellation}. The constellation of nodes, formed by VSF, has fundamental difference with traditional clustering technique, where sensors are clustered based on their geographical collocation. VSF, however, can put any two nodes into the same constellation if they are highly correlated irrespective of their geographical location. One representative node generating data from a constellation is enough to predict the data of other nodes in that constellation. Therefore, VSF selects $m$ nodes to be active if there are $m$ constellations while keeping the remaining nodes as dormant. However, VSF fails if these active nodes are isolated i.e., they are not in the communication range of the sink and a multihop path to the sink cannot be built due to many dormant nodes. 

For $m$ constellations, there can be exponentially many combinations of active nodes, each with at least $m$ active nodes. The immediate challenge is to select a particular combination of active nodes representing $m$ constellations. We aim to select the one that requires least amount of energy to report their data to the sink. This is to reduce the overall energy consumption of the WSN in order to increase the network lifetime.  A combination of active nodes with more than $m$ active nodes, but least energy consuming should be preferred over a combination with lesser number of active nodes but consuming more energy. Therefore, it is required to find a set of active nodes that requires minimum energy to report their data to the sink, which also satisfies the following conditions -- (a)~they form a connected graph, i.e., all of them can reach the sink node; (b)~every constellation has at least one representative sensing node so that the data for all the dormant nodes within that constellation can be predicted with a tolerable error bound. We term this optimization problem as \emph{Sleep-Route}.

We formulate the problem with following considerations: (i)~All sensor nodes are homogeneous, i.e., they all have same transmission range and same energy consumption for various operations (sensing, transmit, receive, etc.,) in active/dormant mode.
(ii)~The network is denoted as a node-weighted undirected graph $G=(S,L)$, where $S$ is the set of all sensor nodes in the network and $L$ is the set of existing links. 
(iii)~The nodes are associated with non-zero cost (node-weighted), which is defined by their energy consumption based on their mode of operation. If node $i$ and $j$ are in active and dormant mode, their energy consumptions are denoted as $e_{a}(i)$ and $e_{d}(j)$ respectively.
(iv)~A link between node $i$ and node $j$ exists if they are within their transmission range. The link cost is defined by the energy consumption for transmitting one packet from node $i$ to $j$, and denoted by $e_{l}(i,j)$. For simplicity, we assume that $e_{l}(i,j) = e_{l}(j,i)$.
With this background, we now provide a formal definition of the sleep-route problem.

\begin{definition} Sleep-Route problem: Given a node-weighted undirected graph $G=(S,L)$ and a collection of $m$ sets $S_{i}$ such that $S_{i}\bigcap S_{j}=\phi,\;\forall i\neq j$ and $\bigcup_{i=1:m} S_{i}=S$, find the minimum cost subgraph $G'=(A,L')$. That is, $\sum_{i\in A} e_{a}(i) + \sum_{(i,j)\in L'} e_{l}(i,j)$ is minimum, such that (i)~$A$ is connected, and (ii)~$A$ includes at least one node from each $S_{i}$, i.e., $S_{i}\bigcap A \ne \phi,\ \ \forall i$.
\end{definition}

\section{Time-Complexity and Solution of Sleep-Route Problem}
\label{sec:sr}
In this section, first, we discuss the time-complexity of the problem. We found that the Sleep-Route problem is NP-hard. To find a optimal solution for the problem, we formulate the problem as an Integer Linear Program (ILP). Then, we provide a heuristic algorithm, which finds a near optimal solution in polynomial time. Let's first prove the NP-hardness of the problem.

\begin{theorem}\label{Thm1}
Sleep-Route (decision) problem is NP-complete. The optimization version of Sleep-Route is NP-hard. 
\end{theorem}

\begin{proof}
It is easy to show that Sleep-Route $\in$ NP, since a nondeterministic algorithm needs only to find a subgraph $A$ and then verify in polynomial time that (i)~it includes at least one node from each constellation (complexity is O(n)), and (ii)~it is connected (complexity of DFS is O(n)). Therefore, Sleep-Route (decision) problem is NP.

To show that Sleep-Route is NP-hard, we reduce the problem into the node-weighted Group Steiner Tree (GST) problem. The classical Steiner Tree problem is defined as:

\begin{definition}[Steiner Tree (ST) problem] Given a graph $G=(S,L)$ with all links have non-zero cost (weight) and a set $A\subseteq S$, find the minimum-cost tree that spans nodes in $A$.
\end{definition}

This is one of the first problem shown to be NP-hard by Karp \cite{karp1972reducibility}. A Group Steiner Tree is a generalized version of the Steiner tree problem.

\begin{definition}[Group Steiner Tree (GST) problem] Given a graph $G=(S,L)$ with non-zero link cost, and a collection $S_{1}, S_{2}, ..., S_{m}$ of node sets called groups. Find a minimum-cost connected subgraph of $G$ that contains at least one node from each group.
\end{definition}
Being a generalized version of the Steiner Tree problem, the Group Steiner Tree (GST) problem is also NP-hard \cite{ihler1999class}. A general node-weighted Steiner Tree problem is an extension of the classical Steiner Tree problem, where nodes have non-zero cost. 

We map the WSN in Sleep-Route into a graph in GST. For each sensor node $s \in S$, they are the nodes in the graph in GST. The constellations in Sleep-Route are the groups of nodes in GST. The links and their respective weights are calculated based on the energy consumed to deliver one packet through that link, thus, are the same for both Sleep-Route and GST. Similarly, the node weights are calculated based on the energy consumed by a node for being active and sensing. Thus, the transformed graph becomes a graph in Sleep-Route problem. There is a one-to-one mapping between the Sleep-Route problem and the node-weighted GST problem. As a result, a ``yes-instance'' of GST, will always reply ``yes'' for Sleep-Route. Similarly, Sleep-Route will never reply a ``yes'' to a ``no-instance'' in GST. Therefore, the decision version of Sleep-Route is Np-complete as this is the case for GST. By default, the optimization problem then becomes a NP-hard problem.
\end{proof}

\subsection{Integer Programming Formulation of Sleep-Route}
As the \emph{Sleep-Route} problem is NP-hard, we propose a heuristic algorithm. To show how well the heuristic performs, we try to find an optimal solution using Integer Linear Programming (ILP) for smaller problems. Though ILP provides the optimal set of active nodes and links, ILP itself is NP-hard \cite{karp1972reducibility}. To formulate Sleep-Route as an ILP problem, we define a couple of decision variables. 

The operational mode of any node $i$ is denoted by $x_{a}(i) =\{0,1\}$, where 1 implies that the sensor node $i$ is active and 0 for when it is dormant. We represent the existing links in the network as an adjacency matrix ($\boldsymbol{L}$). If there exists a link between node $i$ and $j$, then $l(i,j)=1$; $l(i,j)=0$ otherwise. To decide inclusion (or exclusion) of a link in the optimal subgraph, we use a decision matrix $\boldsymbol{X_{l}}$, where each entry is associated with the corresponding element of the adjacency matrix. As the adjacency matrix is symmetric, we use only the upper triangular part of the decision matrix.

\begin{figure*}
\centering
	\begin{subfigure}[b]{0.2\linewidth}
		\includegraphics[width=\textwidth]{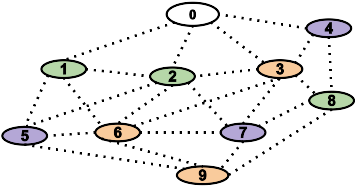}
		\caption{}
		\label{fig:const}
	\end{subfigure}
\hspace{2em}
	\begin{subfigure}[b]{0.2\linewidth}
		\includegraphics[width=\textwidth]{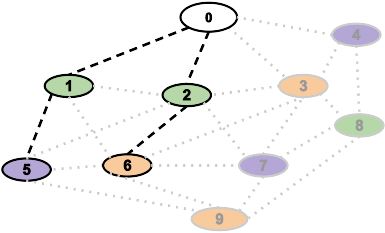}
		\caption{}
		\label{fig:valid_active}
	\end{subfigure}
\hspace{2em}
	\begin{subfigure}[b]{0.2\linewidth}
		\includegraphics[width=\textwidth]{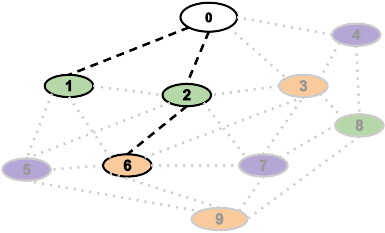}
		\caption{}
		\label{fig:invalid_active}
	\end{subfigure}
\hspace{2em}
	\begin{subfigure}[b]{0.2\linewidth}
		\includegraphics[width=\textwidth]{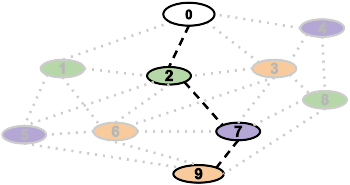}
		\caption{}
		\label{fig:opt_active}
	\end{subfigure}
\caption{(a) A WSN with 9 sensing nodes having three constellations of sensor nodes (marked with different colors); (b) a set of active nodes, but not every constellations are represented; (b) a set of connected active nodes representing all the constellations; (d) another set with minimum number of connected active nodes, which also represents all the constellations.} 
\label{fig:network}
\end{figure*}
\begin{table*}
\renewcommand{\arraystretch}{1.2}
%\center
\centering
\caption{Constellation Matrix and its transformations.}
\begin{subtable}[b]{0.45\linewidth}
\hspace{4em}
	\begin{tabular}{c|c|c|c|c|c|c|c|c|c|c|}
	\multicolumn{1}{c}{}  &  \multicolumn{1}{c}{0} & \multicolumn{1}{c}{1} & 
	\multicolumn{1}{c}{2} &  \multicolumn{1}{c}{3} & \multicolumn{1}{c}{4} & 
	\multicolumn{1}{c}{5} &  \multicolumn{1}{c}{6} & \multicolumn{1}{c}{7} & 
	\multicolumn{1}{c}{8} &  \multicolumn{1}{c}{9} \\ \cline{2-11}
	0 & 1 & 0 & 0 & 0 & 0 & 0 & 0 & 0 & 0 & 0 \\ \cline{2-11}
	1 & 0 & 1 & 1 & 0 & 0 & 0 & 0 & 0 & 1 & 0 \\ \cline{2-11}
	2 & 0 & 1 & 1 & 0 & 0 & 0 & 0 & 0 & 1 & 0 \\ \cline{2-11}
	3 & 0 & 0 & 0 & 1 & 0 & 0 & 1 & 0 & 0 & 1 \\ \cline{2-11}
	4 & 0 & 0 & 0 & 0 & 1 & 1 & 0 & 1 & 0 & 0 \\ \cline{2-11}
	5 & 0 & 0 & 0 & 0 & 1 & 1 & 0 & 1 & 0 & 0 \\ \cline{2-11}
	6 & 0 & 0 & 0 & 1 & 0 & 0 & 1 & 0 & 0 & 1 \\ \cline{2-11}
	7 & 0 & 0 & 0 & 0 & 1 & 1 & 0 & 1 & 0 & 0 \\ \cline{2-11}
	8 & 0 & 1 & 1 & 0 & 0 & 0 & 0 & 0 & 1 & 0 \\ \cline{2-11}
	9 & 0 & 0 & 0 & 1 & 0 & 0 & 1 & 0 & 0 & 1 \\ \cline{2-11}
	\multicolumn{1}{c}{} & \multicolumn{1}{c}{} & \multicolumn{1}{c}{} & \multicolumn{1}{c}{} & \multicolumn{1}{c}{} & \multicolumn{1}{c}{} & \multicolumn{1}{c}{} & \multicolumn{1}{c}{} \\
	\end{tabular}	
	\caption{}
	\label{table:mat1}
\end{subtable}
\hspace{2em}
\begin{subtable}[b]{0.45\linewidth}
	\begin{tabular}{c|c|c|c|c|c|c|c|c|c|c|}
	\multicolumn{1}{c}{}  &  \multicolumn{1}{c}{0} & \multicolumn{1}{c}{1} & 
	\multicolumn{1}{c}{2} &  \multicolumn{1}{c}{3} & \multicolumn{1}{c}{4} & 
	\multicolumn{1}{c}{5} &  \multicolumn{1}{c}{6} & \multicolumn{1}{c}{7} & 
	\multicolumn{1}{c}{8} &  \multicolumn{1}{c}{9} \\ \cline{2-11}
	{\color{red} $x_{a}(0)=1$} & {\color{red} 1} & {\color{red} 0} & {\color{red} 0} & {\color{red} 0} & {\color{red} 0} & {\color{red} 0} & {\color{red} 0} & {\color{red} 0} & {\color{red} 0} & {\color{red} 0} \\ \cline{2-11}
	{\color{red} $x_{a}(1)=1$} & {\color{red} 0} & {\color{red} 1} & {\color{red} 1} & {\color{red} 0} & {\color{red} 0} & {\color{red} 0} & {\color{red} 0} & {\color{red} 0} & {\color{red} 1} & {\color{red} 0} \\ \cline{2-11}
	{\color{red} $x_{a}(2)=1$} & {\color{red} 0} & {\color{red} 1} & {\color{red} 1} & {\color{red} 0} & {\color{red} 0} & {\color{red} 0} & {\color{red} 0} & {\color{red} 0} & {\color{red} 1} & {\color{red} 0} \\ \cline{2-11}
	$x_{a}(3)=0$ & 0 & 0 & 0 & 1 & 0 & 0 & 1 & 0 & 0 & 1 \\ \cline{2-11}
	$x_{a}(4)=0$ & 0 & 0 & 0 & 0 & 1 & 1 & 0 & 1 & 0 & 0 \\ \cline{2-11}
	{\color{red} $x_{a}(5)=1$} & {\color{red} 0} & {\color{red} 0} & {\color{red} 0} & {\color{red} 0} & {\color{red} 1} & {\color{red} 1} & {\color{red} 0} & {\color{red} 1} & {\color{red} 0} & {\color{red} 0} \\ \cline{2-11}
	{\color{red} $x_{a}(6)=1$} & {\color{red} 0} & {\color{red} 0} & {\color{red} 0} & {\color{red} 1} & {\color{red} 0} & {\color{red} 0} & {\color{red} 1} & {\color{red} 0} & {\color{red} 0} & {\color{red} 1} \\ \cline{2-11}
	$x_{a}(7)=0$ & 0 & 0 & 0 & 0 & 1 & 1 & 0 & 1 & 0 & 0 \\ \cline{2-11}
	$x_{a}(8)=0$ & 0 & 1 & 1 & 0 & 0 & 0 & 0 & 0 & 1 & 0 \\ \cline{2-11}
	$x_{a}(9)=0$ & 0 & 0 & 0 & 1 & 0 & 0 & 1 & 0 & 0 & 1 \\ \cline{2-11}
	\multicolumn{1}{c}{} & \multicolumn{1}{c}{{\color{blue}1}} & \multicolumn{1}{c}{{\color{blue}2}} & \multicolumn{1}{c}{{\color{blue}2}} & \multicolumn{1}{c}{{\color{blue}1}} & \multicolumn{1}{c}{{\color{blue}1}} & \multicolumn{1}{c}{{\color{blue}1}} & \multicolumn{1}{c}{{\color{blue}1}} & \multicolumn{1}{c}{{\color{blue}1}} & \multicolumn{1}{c}{{\color{blue}2}} & \multicolumn{1}{c}{{\color{blue}1}} \\
	\end{tabular}	
	\caption{}
	\label{table:mat2}
\end{subtable}
\hspace{2em}
\begin{subtable}[b]{0.45\linewidth}
	\begin{tabular}{c|c|c|c|c|c|c|c|c|c|c|}
	\multicolumn{1}{c}{}  &  \multicolumn{1}{c}{0} & \multicolumn{1}{c}{1} & 
	\multicolumn{1}{c}{2} &  \multicolumn{1}{c}{3} & \multicolumn{1}{c}{4} & 
	\multicolumn{1}{c}{5} &  \multicolumn{1}{c}{6} & \multicolumn{1}{c}{7} & 
	\multicolumn{1}{c}{8} &  \multicolumn{1}{c}{9} \\ \cline{2-11}
	{\color{red} $x_{a}(0)=1$} & {\color{red} 1} & {\color{red} 0} & {\color{red} 0} & {\color{red} 0} & {\color{red} 0} & {\color{red} 0} & {\color{red} 0} & {\color{red} 0} & {\color{red} 0} & {\color{red} 0} \\ \cline{2-11}
	{\color{red} $x_{a}(1)=1$} & {\color{red} 0} & {\color{red} 1} & {\color{red} 1} & {\color{red} 0} & {\color{red} 0} & {\color{red} 0} & {\color{red} 0} & {\color{red} 0} & {\color{red} 1} & {\color{red} 0} \\ \cline{2-11}
	{\color{red} $x_{a}(2)=1$} & {\color{red} 0} & {\color{red} 1} & {\color{red} 1} & {\color{red} 0} & {\color{red} 0} & {\color{red} 0} & {\color{red} 0} & {\color{red} 0} & {\color{red} 1} & {\color{red} 0} \\ \cline{2-11}
	$x_{a}(3)=0$ & 0 & 0 & 0 & 1 & 0 & 0 & 1 & 0 & 0 & 1 \\ \cline{2-11}
	$x_{a}(4)=0$ & 0 & 0 & 0 & 0 & 1 & 1 & 0 & 1 & 0 & 0 \\ \cline{2-11}
	$x_{a}(5)=0$ & 0 & 0 & 0 & 0 & 1 & 1 & 0 & 1 & 0 & 0 \\ \cline{2-11}
	{\color{red} $x_{a}(6)=1$} & {\color{red} 0} & {\color{red} 0} & {\color{red} 0} & {\color{red} 1} & {\color{red} 0} & {\color{red} 0} & {\color{red} 1} & {\color{red} 0} & {\color{red} 0} & {\color{red} 1} \\ \cline{2-11}
	$x_{a}(7)=0$ & 0 & 0 & 0 & 0 & 1 & 1 & 0 & 1 & 0 & 0 \\ \cline{2-11}
	$x_{a}(8)=0$ & 0 & 1 & 1 & 0 & 0 & 0 & 0 & 0 & 1 & 0 \\ \cline{2-11}
	$x_{a}(9)=0$ & 0 & 0 & 0 & 1 & 0 & 0 & 1 & 0 & 0 & 1 \\ \cline{2-11}
	\multicolumn{1}{c}{} & \multicolumn{1}{c}{{\color{blue}1}} & \multicolumn{1}{c}{{\color{blue}2}} & \multicolumn{1}{c}{{\color{blue}2}} & \multicolumn{1}{c}{{\color{blue}1}} & \multicolumn{1}{c}{{\color{blue}0}} & \multicolumn{1}{c}{{\color{blue}0}} & \multicolumn{1}{c}{{\color{blue}1}} & \multicolumn{1}{c}{{\color{blue}0}} & \multicolumn{1}{c}{{\color{blue}2}} & \multicolumn{1}{c}{{\color{blue}1}} \\
	\end{tabular}	
	\caption{}
	\label{table:mat3}
\end{subtable}
\hspace{2em}
\begin{subtable}[b]{0.45\linewidth}
	\begin{tabular}{c|c|c|c|c|c|c|c|c|c|c|}
	\multicolumn{1}{c}{}  &  \multicolumn{1}{c}{0} & \multicolumn{1}{c}{1} & 
	\multicolumn{1}{c}{2} &  \multicolumn{1}{c}{3} & \multicolumn{1}{c}{4} & 
	\multicolumn{1}{c}{5} &  \multicolumn{1}{c}{6} & \multicolumn{1}{c}{7} & 
	\multicolumn{1}{c}{8} &  \multicolumn{1}{c}{9} \\ \cline{2-11}
	{\color{red} $x_{a}(0)=1$} & {\color{red} 1} & {\color{red} 0} & {\color{red} 0} & {\color{red} 0} & {\color{red} 0} & {\color{red} 0} & {\color{red} 0} & {\color{red} 0} & {\color{red} 0} & {\color{red} 0} \\ \cline{2-11}
	$x_{a}(1)=0$ & 0 & 1 & 1 & 0 & 0 & 0 & 0 & 0 & 1 & 0 \\ \cline{2-11}
	{\color{red} $x_{a}(2)=1$} & {\color{red} 0} & {\color{red} 1} & {\color{red} 1} & {\color{red} 0} & {\color{red} 0} & {\color{red} 0} & {\color{red} 0} & {\color{red} 0} & {\color{red} 1} & {\color{red} 0} \\ \cline{2-11}
	$x_{a}(3)=0$ & 0 & 0 & 0 & 1 & 0 & 0 & 1 & 0 & 0 & 1 \\ \cline{2-11}
	$x_{a}(4)=0$ & 0 & 0 & 0 & 0 & 1 & 1 & 0 & 1 & 0 & 0 \\ \cline{2-11}
	$x_{a}(5)=0$ & 0 & 0 & 0 & 0 & 1 & 1 & 0 & 1 & 0 & 0 \\ \cline{2-11}
	$x_{a}(6)=0$ & 0 & 0 & 0 & 1 & 0 & 0 & 1 & 0 & 0 & 1 \\ \cline{2-11}
	{\color{red} $x_{a}(7)=1$} & {\color{red} 0} & {\color{red} 0} & {\color{red} 0} & {\color{red} 0} & {\color{red} 1} & {\color{red} 1} & {\color{red} 0} & {\color{red} 1} & {\color{red} 0} & {\color{red} 0} \\ \cline{2-11}
	$x_{a}(8)=0$ & 0 & 1 & 1 & 0 & 0 & 0 & 0 & 0 & 1 & 0 \\ \cline{2-11}
	{\color{red} $x_{a}(9)=1$} & {\color{red} 0} & {\color{red} 0} & {\color{red} 0} & {\color{red} 1} & {\color{red} 0} & {\color{red} 0} & {\color{red} 1} & {\color{red} 0} & {\color{red} 0} & {\color{red} 1} \\ \cline{2-11}
	\multicolumn{1}{c}{} & \multicolumn{1}{c}{{\color{blue}1}} & \multicolumn{1}{c}{{\color{blue}1}} & \multicolumn{1}{c}{{\color{blue}1}} & \multicolumn{1}{c}{{\color{blue}1}} & \multicolumn{1}{c}{{\color{blue}1}} & \multicolumn{1}{c}{{\color{blue}1}} & \multicolumn{1}{c}{{\color{blue}1}} & \multicolumn{1}{c}{{\color{blue}1}} & \multicolumn{1}{c}{{\color{blue}1}} & \multicolumn{1}{c}{{\color{blue}1}} \\
	\end{tabular}	
	\caption{}
	\label{table:mat4}
\end{subtable}
\caption{(a) Original constellation matrix based on high correlation among nodes; (b-d) Transformed matrix according to active set of nodes shown in Fig.~\ref{fig:valid_active}, \ref{fig:invalid_active} and \ref{fig:opt_active} respectively.}
\label{table:const_matrix}
\end{table*}

A cost is associated with each link. The link cost between nodes $i$ and $j$ is set as the energy required to transmit one packet from node $i$ to node $j$ and is calculated as, 
\newcommand{\twopartdef}[4]
{
	\left\{
		\begin{array}{ll}
			#1 & #2 \\
			#3 & #4
		\end{array}
	\right.
}
\begin{equation}
e_{l}(i,j) = \twopartdef { e_{tx}(i) + e_{rv}(j), } {\mbox{if $l(i,j)=1$};} {\infty} {\mbox{otherwise},}
\label{eq:cost}
\end{equation}
where $e_{tx}(i)$ and $e_{rv}(j)$ are the energy cost for node $i$ and $j$ to transmit and receive a packet respectively. Since the sensor nodes are homogeneous and they are assumed to send packet at the same transmit power, the energy cost for each link is the same. However, wireless links are not homogeneous, and the link quality varies due to varying packet loss, interference, etc. As a result, packets need to be retransmitted from the sender. We try to accommodate this aspect in our design. If the packet delivery probability between node $i$ and $j$ is denoted as $P_{d}(i,j)$, then the link cost in (\ref{eq:cost}) is modified as,
\begin{equation}
e_{l}(i,j) = \twopartdef {[e_{tx}(i) + e_{rv}(j)]/P_{d}(i,j), } {\mbox{if $l(i,j)=1$};} {\infty} {\mbox{otherwise}.}
\label{eq:linkcost}
\end{equation}
We assume that there is only one sink node ($s_{0}$) in the WSN and it is a special node with unlimited or sufficiently enough resources. The sink node is always active, i.e., $x_{a}(0)=1$ to receive the data. Thus, its energy consumption is not considered while minimizing the energy consumption of the whole network. The constellations of nodes are represented using a matrix. Suppose, $\boldsymbol{C}$ is the constellation matrix. At the end of a training period, VSF creates the constellation matrix, where $c(i,j)$ (and also $c(j,i)$)) is associated with the correlation between the sensed data of nodes $i$ and $j$. If the correlation is above certain threshold, the nodes are assigned to the same constellation i.e., $c(i,j)=1$, and $c(i,j)=0$ otherwise. VSF provides the constellations of sensor nodes available in a WSN. The sink node forms a constellation containing only itself. A small WSN and its constellations are shown in Fig.~\ref{fig:const} and Table~\ref{table:mat1} respectively.

The optimal subgraph, i.e., the set of active nodes ($A$) and also the links connecting them to the sink ($L'$), need to be changed over time. Otherwise, one set of active nodes will drain their energy, and these nodes will become unavailable soon. As a result, the network might become disconnected. To avoid such a situation, a better approach would be to balance the energy consumption amongst all the nodes within the network. To achieve this, the nodes with more residual energy should be preffered to be selected as active node. Thus, we calculate the weight of a sensor node $i$ as,
\begin{equation}
e_{a}(i) = e_{active} + e_{sense} + (E_{res} - e_{res}(i)),
\label{eq:nodecost}
\end{equation}
where $e_{active}$ is the energy required for any node to be active, $e_{res}(i)$ is the residual energy of $i$, and $E_{res}$ is the maximum residual energy over all the sensor nodes in the network ($\max_{k \in S}\{{e_{res}(k)}\}$). Clearly, the sensor node with more residual energy will have smaller value of $e_{a}$ (bacause of smaller value of $E_{res} - e_{res}(i)$). As a result, it will be preferred to be selected over a sensor node with relatively less residual energy, i.e. larger $e_{a}$. At the beginning of every operational period (see Fig.~\ref{fig:vs_train}), a new optimal subgraph (new set of active nodes and links) are selected. By incorporation the residual energy of a sensor node in calculating its active cost, we ensure fairness in energy comnsumption among the sensor nodes.

Now, considering there is total $n+1$ nodes in the WSN including the sink node, the minimization problem for Sleep-Route is defined as, 

\begin{alignat}{2}
\min & \sum_{i=1}^{n} e_{a}(i)\; x_{a}(i) + \sum_{i=0}^{n-1} \sum_{j=i+1}^{n} e_{l}(i,j)\; x_{l}(i,j) \label{eq:min_func}
\end{alignat}
subject to:
\begin{alignat}{3}
%x_{a}(0) & = 1 & \nonumber  \\
\sum_{i=0}^{n-1} \sum_{j=i+1}^{n} x_{l}(i,j) & = \sum_{i=0}^{n} x_{a}(i) - 1 \label{eq:tree} \\
\sum_{i=0}^{n-1} \sum_{j=i+1}^{n} x_{l}(i,j)v(i)v(j) & \leq \sum_{i=0}^{n} x_{a}(i)v(i) - 1, \label{eq:noloop}
\end{alignat}
for all possible binary vector ($\boldsymbol{v}$) of length $(n+1)$.
\begin{alignat}{3}
2x_{l}(i,j) & \leq x_{a}(i) + x_{a}(j), & ~~~~ \forall i, \; \forall j \label{eq:endpoints} \\
\sum_{i=0}^{n} x_{a}(i) c(i,j) & \geq 1, & ~~~~ \forall j \label{eq:representation}
\end{alignat}

To maintain connectivity, for each node in $A$, there should be at least one active link adjacent to it. In fact the minimum number of links that connects a set of nodes forms a tree. If there are $k$ nodes in $A$, then there should be $k-1$ links in $L'$. Constraint (\ref{eq:tree}) ensures this cardinality condition. But $k-1$ links does not always ensure connectivity (in case of cycles). To remove any cycle from the graph, all the sub-tours need to be eliminated. If $A'$ represents any subset of sensor nodes and $E(A')$ represents the links among the sensor nodes in $A'$, then $\sum_{e\in E(A')} e \leq |A'| - 1$ ensures no possible sub-tour. To select a subset of sensor nodes, a vector ($|\boldsymbol{v}|=n+1$) is used, where $v(i)=1$ represents $i\in A'$, and $v(i)=0$ otherwise. Considering $v(0)=1$ (for the sink node), there can be $2^n$ possible vectors representing all possible subsets of sensor nodes, and for each subset, sub-tour needs to be eliminated (constraint (\ref{eq:noloop})).

The requirement of forming a connected tree amongst selected active nodes cannot be guaranteed by constraints (\ref{eq:tree}) and (\ref{eq:noloop}). To ensure that all the active nodes are connected, the end points of every selected link, i.e., both the nodes connected via the link have to be in the active set of nodes. The constraint (\ref{eq:endpoints}) ensures that if $(i,j) \in L'$, then its end points are also in $A$. However, the opposite is not true, i.e., even if a sensor node is in $A$, all of its existing links may not be in $L'$.  

The final constraint (\ref{eq:representation}) ensures that every constellation has at least one active node as representative. However, in some cases more than one node can be selected to be active in order to maintain connectivity. Let us take an example to clarify whether (\ref{eq:representation}) can satisfy the criteria. Suppose, $\boldsymbol{x_{a}}=[1110011000]^{T}$ represents the decision vector for sensing nodes (Fig.~\ref{fig:valid_active}). According to the L.H.S. of (\ref{eq:representation}), this decision vector transforms the constellation matrix as shown in Table~\ref{table:mat2}. As the sum of the elements of each columns are greater than 1, this decision vector fulfills the criteria of representing each constellation by at least one active node. On the other hand, a decision vector like $x_{a}=[1110001000]^{T}$, cannot fulfill the criteria (shown in Table~\ref{table:mat3}). As there is no representation from the constellation with sensor node 4, 5 and 7, it becomes an invalid decision vector (Fig.~\ref{fig:invalid_active}). A better choice of active set of nodes is shown in Fig.~\ref{fig:opt_active}, which also satisfies the representation criteria (Table~\ref{table:mat4}).

\subsection{Solution by Heuristic}
As the number of constraints (\ref{eq:noloop}) grow exponentially with the number of sensor nodes, a solution by ILP is not feasible. To solve the Sleep-Route problem in polynomial time, we propose a heuristic algorithm. The heuristic algorithm uses approximation at two stages. At first, we reduce the general graph into a tree. Then, we select a set of active nodes based on some criteria such that all the constellations are represented. During the node selection process, we might select more than one active node from each constellation to maintain connectivity. If a node is selected as active, its parent node must be active to ensure connectivity. The child-parent relationship is derived from the tree formed during the first phase of the heuristic.

To find a tree, we have used the Dijkstra's algorithm, where a shortest path is found from any sensor node to the sink node. By shortest path, we refer to the minimum costl path to send a packet from the sensor node to the sink node. The shortest path algorithm removes some of the links from the adjacency matrix (i.e., the link cost is set to $\infty$) and ensures that every node has exactly one parent node.

After forming the tree, a set of active nodes is selected (based on some criteria) representing all the constellations of the network. Additional sensor nodes can be selected such that every node's parent node is also in the active set of nodes. Then, all such links are selected for whom both the end points (sensor nodes) are selected to be active. This ensures all the constraints [\ref{eq:tree}-\ref{eq:representation}]. The optimality depends on the criteria to select the active node.

We first employ a randomized selection of active nodes (Algorithm~\ref{algo:rand1}). The algorithm takes the global set of sensor nodes ($S$) and partition them into two subsets - active ($A$) and dormant ($D$). First, the algorithm removes the sink node from the global set ($S$) and adds it to the active set ($A$) (line 2,3). Then, the algorithm starts removing a sensing node ($s$) randomly from the remaining sensor nodes of the global set. The recently selected node is then added to the active set. To ensure connectivity, its parent node (also the parent's parent) is also need to be assigned to the active set. We define a recursive algorithm to decide a node's state (Algorithm~\ref{algo:marknode}). First, the selected node is removed from $S$ and assigned to $A$ (line 1-2). Then we find the parent node of $s$ (line 3). If the parent node is not already assigned to the active set, we recursively call the algorithm to assign it to the active set (line 4-6). Next, all the sensor nodes that belongs to the same constellation with $s$, are removed from the global set and added to the dormant set (line 7-12). The algorithm (Algorithm~\ref{algo:rand1}) continues, until the global set becomes empty.

\begin{algorithm}
\caption{A randomized algorithm for selecting a subset sensor nodes as active.}
\label{algo:rand1}
\begin{algorithmic}[1]
\Statex \textbf{Input:} Set of all sensors in the network ($S$), set of all existing link ($L$), the sink node ($s_{0}$), and the constellation matrix ($C$).
\Statex \textbf{Output:} Two set of nodes - Active ($A$) and Dormant ($D$).
\State $[S,\hat{L}] \leftarrow shortestPathTree(S,L)$;
\State $S \leftarrow S - s_{0}$;
\State $A \leftarrow A + s_{0}$;
\State $D \leftarrow {\phi}$;
\While{isNotEmpty($S$)} 
	\State randomly select a sensor node $s$ from $S$;
	\State $markNodeState(s)$;
\EndWhile
\end{algorithmic}
\end{algorithm}

\begin{algorithm}
\caption{\emph{markNodeState:} It assigns a sensor node (and its parent) to the active set($A$), and all the sensor nodes that belongs to the same constellation to the dormant set ($D$).}
\label{algo:marknode}
\begin{algorithmic}[1]
\State $S \leftarrow S - s$;
\State $A \leftarrow A + s$;
\State $s' \leftarrow parent(s)$;
\If{ $s'$ is in $S$ }
	\State $markNodeState(s')$;
\EndIf
\For {all $\hat{s}=1:n$}
	\If{$c(s,\hat{s}) = 1$ AND $\hat{s}$ is in $S$} 
		\State $S \leftarrow S - \hat{s}$;
		\State $D \leftarrow D + \hat{s}$;
	\EndIf 
\EndFor
\end{algorithmic}
\end{algorithm}

Even though the randomized algorithm tries to select one active node per constellation and ensures availability of a path from each sensing node to the sink node, it does not guarantee the minimum cost active set of nodes. There are two reasons - first, the random selection of sensing nodes imposes more number of active nodes (to maintain connectivity); second, the optimal set of nodes, i.e., the sensor nodes with minimum active cost might not be selected. To solve these issues, we employ a deterministic algorithm that follows a greedy approach. 

Instead of selecting the sensing nodes randomly, the greedy algorithm uses selection probability for each sensor node based on their active cost (See~\ref{eq:nodecost} for active cost calculation). The node with minimum active cost is assigned the highest selection probability. If a node has the highest selection probability among the remaining nodes in $S$, it is assigned to the active set. The node assignment is the same as the previous (Algorithm~\ref{algo:marknode}). As the heuristic algorithm does not consider the link cost and node cost together, while assigning the state, an optimal solution might not be found. To find the possible worst-case performance of the heuristic, we have tried another version of the node selection criteria, where a node with lowest selection probability is assigned to the active set. The number of active nodes selected using the deterministic algorithm is lower than for the randomized algorithm. The results shown in the next section confirms this and the overall energy consumption is much closer to the integer programming solution as compared to the randomized algorithm.

\section{Evaluation}
\label{sec:eva}
Various relevant parameters and energy costs for a Tmote sky node are listed in Table~\ref{settings}.
\begin{table}
\centering
\caption{System Parameters and Settings}
\begin{tabular}{|l|l|} \hline
\textbf{Parameter} & \textbf{Value} \\ \hline
Initial energy at every sensor node & 100\,J \\ \hline
Message size & 128\,B \\ \hline
Transmission power & 5\,dBm \\ \hline
Energy cost for sending a message & 450\,$\mu$J \\ \hline
Energy cost for receiving a message & 400\,$\mu$J \\ \hline
Energy cost for sensing temperature & 20.30\,$\mu$J \\ \hline
Energy cost in active mode & 4.898\,mW \\ \hline
Energy cost in sleeping mode & 0.144\,mW \\
\hline\end{tabular}
\label{settings}
\end{table}
In our evaluation, we measured energy consumption by all the sensor nodes during the operational period (Fig.~\ref{fig:vs_train}). As mentioned earlier, a new subset of active nodes are selected before every operational period. We have fixed the length of an operational period to 20 samples. Sensor nodes collected data every 31\,s; thus, an operational period lasts for 620\,s. Even though we did not assume any particular MAC protocol, the duty cycling is fixed to 5\% for the active nodes. So, if a node is assigned as active, the minimum energy consumption by the sensor node for the period is: $T_{p}*d_{cycle}*e_{active} + T_{p}*(1-d_{cycle})*e_{sleep} = (620*0.05*4.898 + 620*0.95*0.144) \;mJ =\; 236.654 \;\;mJ$,
where $T_{p}$ is the duration of the operational period, $d_{cycle}$ is the duty-cycling of active nodes, $e_{active}$ is the energy required for a node to be active, and $e_{sleep}$ is the energy spent by a node in dormant mode. Please note that this cost does not include any communication and sensing cost. In contrast, a sensor node in dormant mode consumes only $(620*0.144) \;mJ = 89.28 \;mJ$ during the operational period.

The existence of an edge between two nodes is decided based on the link quality between two nodes. The link quality is decided based on the successful packet delivery probability ($p_d$). Thus, we require $1/p_d$ transmissions to deliver a packet successfully. We have used $p_d$ to find the cost of an edge to decide the shortest path tree (See~\ref{eq:linkcost}).
\begin{figure}[]
	\centering
	\includegraphics[width=\linewidth]{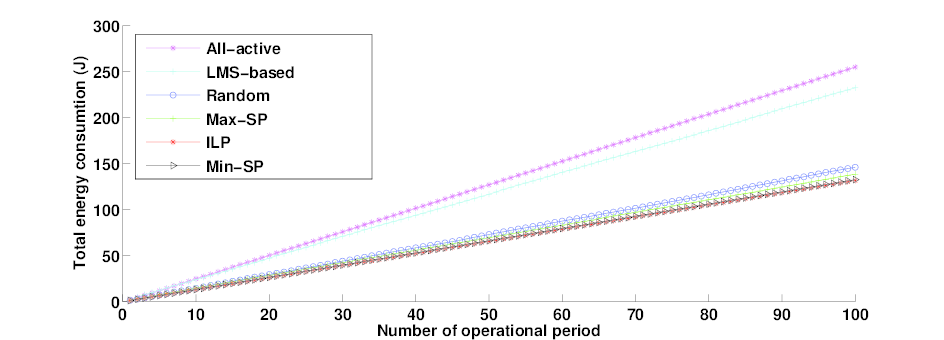}
	\caption{Total energy consumption by all the sensor nodes in the network for various data collection methods.}
	\label{fig:e_spent}
\end{figure}
\begin{figure}[]
	\centering
	\includegraphics[width=\linewidth]{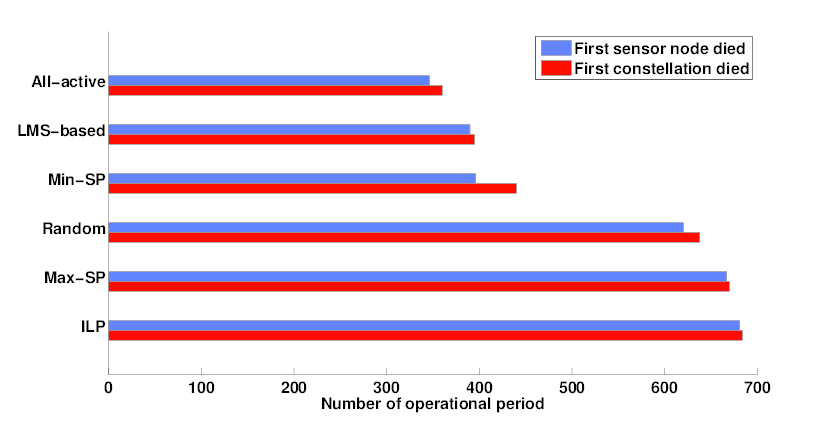}
	\caption{Lifetime of the WSN for various data collection schemes.}
	\label{fig:lifetime}
\end{figure}

Lifetime of a WSN can be defined in many ways, e.g., until the first node dies (runs out of battery), until the remaining alive nodes cannot communicate with the sink, etc. In this work, we define a deployment lifetime as `until all the sensor nodes belonging to a particular constellation dies'. Under this assumption, we show the effectiveness of the Sleep-Route scheme in extending lifetime of a WSN.

As discussed earlier, the ILP formulation for Sleep-Route has exponential number of constraints. As a result, when the number of sensor nodes is large, an optimal Sleep-Route solution via ILP is not feasible. Thus, we have evaluated the heuristic algorithm in comparison with the ILP when applied on the WSN shown in Fig.~\ref{fig:network}. First, we ran the experiment for 100 operational periods. The optimal set of active nodes and links are calculated via ILP before every operational period. Similarly, the active set of nodes is calculated based on the heuristic algorithm. As mentioned earlier, after forming the shortest path tree, the active set of nodes is selected based on three criteria - random selection of active nodes (RS), active node selection based on maximum selection probability (Max-SP), and active node selection based on minimum selection probability (Min-SP).  For each of these methods, the total energy consumption in the whole network is calculated. Additionally, we have calculated the energy consumption, (i)~if all the sensor nodes are active and report their data to the sink node by forming a collection tree, and (ii)~data is collected based on the LMS-based method described by Santini {\it et al.}~\cite{santini2006adaptive}. 
 
For all the schemes, the cumulative energy consumption after 100 operational periods is shown in Fig.~\ref{fig:e_spent}. Please note that the LMS-based scheme restricts data transmission by a node if the sensed data can be predicted by the sink with a tolerable error bound. Due to this, we assumed that only 2 out of 9 sensor nodes report their data in each sensing interval.  These 2 nodes are selected randomly. The remaining 7 nodes differ data transmissions, and hence saving in energy. As expected, the Sleep-Route solution (ILP as well as the heuristics) requires less energy than LMS-based method. Out of three heuristic approaches, the Max-SP based method should perform the best. However, it seems that the Min-SP based method performs exceptionally well. Even if the total energy consumption by Min-SP based heuristic is less than Max-SP based heuristic, the lifetime of the WSN is drastically reduced by Min-SP based heuristic (Fig.~\ref{fig:lifetime}).

The Sleep-Route scheme tries to rotate the set of active nodes by considering the residual energy available in each sensor node. Otherwise, some sensor nodes might die early, which eventually may lead to shorter lifetime. The selection probability tries to impose this behavior on the selection of nodes, i.e., the sensor node with more residual energy has higher chance of getting selected. However, selecting nodes with minimum selection probability, i.e., less residual energy can impose more burden on some nodes. As a result the lifetime of a WSN can be decreased in spite of the fact that the overall energy consumption at every operational period is less. Fig.~\ref{fig:lifetime} supports this argument.
\section{Conclusion}
We define the \emph{Sleep-Route} problem that arises in a WSN consists of spatio-temporally correlated sensor nodes. We prove that the problem is NP-hard. Using a polynomial time greedy-approach based heuristic algorithm, we solve the problem near optimally. The solution improves the lifetime of a WSN. In future, we plan to evaluate our algorithm on a real sensor deployment. Further it is interesting to consider situation where some nodes could be powered by mains power\cite{vazifehdan2012analytical}. This is an interesting extension.

%ACKNOWLEDGMENTS are optional
%\section*{Acknowledgments}
%This work was supported by an EU FP7 project, called iCore (contract number: 287708).

% trigger a \newpage just before the given reference
% number - used to balance the columns on the last page
% adjust value as needed - may need to be readjusted if
% the document is modified later
%\IEEEtriggeratref{8}
% The "triggered" command can be changed if desired:
%\IEEEtriggercmd{\enlargethispage{-5in}}

% references section

% can use a bibliography generated by BibTeX as a .bbl file
% BibTeX documentation can be easily obtained at:
% http://www.ctan.org/tex-archive/biblio/bibtex/contrib/doc/
% The IEEEtran BibTeX style support page is at:
% http://www.michaelshell.org/tex/ieeetran/bibtex/
%\bibliographystyle{IEEEtran}
% argument is your BibTeX string definitions and bibliography database(s)
%\bibliography{IEEEabrv,../bib/paper}
%
% <OR> manually copy in the resultant .bbl file
% set second argument of \begin to the number of references
% (used to reserve space for the reference number labels box)
\bibliographystyle{abbrv}
\bibliography{sleeproute}

% that's all folks
\end{document}